\documentclass[submission,copyright,creativecommons]{eptcs}
\providecommand{\event}{SYNT 2015}

\usepackage{macros}
\usepackage{hyperref}
\usepackage{amsthm}
\usepackage{float}
\usepackage{caption}

\newtheorem{theorem}{Theorem}
\newtheorem{proposition}{Proposition}
\newtheorem{lemma}{Lemma}

\title{Compositional Algorithms for Succinct Safety Games}

\author{
Romain Brenguier\thanks{Authors supported by the ERC inVEST (279499) project.},
Guillermo A. P\'{e}rez\thanks{Author supported by F.R.S.-FNRS fellowship.},\\
Jean-Fran\c{c}ois Raskin$^\ast$\hspace{-1.7mm},
Ocan Sankur$^\ast$\\
\email{\{rbrengui,gperezme,jraskin,osankur\}@ulb.ac.be}
\institute{Universit\'{e} Libre de Bruxelles -- Brussels, Belgium}
}

\def\titlerunning{Compositional Algorithms for Succinct Safety Games}
\def\authorrunning{R. Brenguier et al.}

\begin{document}

\maketitle

\begin{abstract}
We study the synthesis of circuits for succinct safety specifications given in
the AIG format.  We show how AIG safety specifications can be decomposed
automatically into sub-specifications.  Then we propose symbolic compositional
algorithms to solve the synthesis problem compositionally starting for the
sub-specifications.  We have evaluated the compositional algorithms on a set of
benchmarks including those proposed for the first synthesis competition
organised in 2014 by the Synthesis Workshop affiliated to the CAV
conference.  We show that a large number of benchmarks can be decomposed
automatically and solved more efficiently with the compositional algorithms that
we propose in this paper.
\end{abstract}

\section{Introduction}
We study the synthesis of circuits for succinct safety specifications given in
the AIG format.  An AIG file for synthesis describes a circuit that compactly
defines a transition relation between valuations for latches, {\em
uncontrollable} and {\em controllable} input signals.  The circuit contains a
special latch called the {\em error latch}.  Initially, all latches are false,
and the controller chooses values for the controllable input signals so as to
always keep the error latch {\em low} (safety objective), no matter how the
environment chooses values for the uncontrollable input signals. The AIG format
is {\em monolithic} in the sense that it is not explicitly structured into
subsystems.  This is unfortunate as in general, complex systems or
specifications are built of smaller sub-parts and taking into account this
structure may be a definite advantage. 

{\em And-Inverter Graphs} (AIG) have been proposed as a way to provide a simple
and compact file format for a model checking competition affiliated to CAV 2007
(see {\tt http://fmv.jku.at/aiger/FORMAT}). This format has been extended to be
the input format for the 2014 {\em reactive synthesis competition}.  Because the
synthesis competition uses the AIG format, and this format is monolithic, all
the tools that took part in the 2014 reactive synthesis competition solved the
synthesis problems {\em monolithically}. Nevertheless, the specifications that
were proposed during the 2014 synthesis competition are, for a large part of
them, generated from higher level descriptions of systems that bear structure.
For example, two of the most interesting sets of benchmarks, {\sf GenBuf} and
{\sf AMBA}, are generated from Reactive(1) specifications (a tractable subset of
LTL specifications)~\cite{BloemJPPS-jcss12}, or directly from LTL
specifications that are conjunctions of smaller LTL sub-formulas. 
   
In this paper, we show that part of the structure lost during the AIG format
translation can be recovered and used to solve the synthesis problem {\em
compositionally}.  First, we propose a static analysis of the AIG file that
returns, when possible, a decomposition of the circuit into smaller sub-circuits
with their own safety specifications. Then we provide three different algorithms
that first solve the sub-games corresponding to the sub-circuits and then
aggregate, following three different heuristics, the results obtained on the
sub-games. Namely, once we have the solution of all the sub-games we aggregate
them by
\begin{inparaenum}[$(i)$]
	\item taking their intersection -- which, we show, over-approximates the
		actual solution of the general game -- and applying the usual
		fixpoint algorithm to it;
	\item assigning a score to each pair of solutions based on the number of
		variables shared and the size of the BDDs obtained after their
		intersection and using said score to aggregate (pair by pair)
		all the solutions;
	\item trying to refine them using information from a single step of the
		fixpoint computation on the general game (\ie projecting the
		resulting ``bad'' states onto each sub-game).
\end{inparaenum}
We have implemented the decomposition, the compositional synthesis
algorithms, and evaluated the approach on the 2014 reactive synthesis
competition benchmarks as well as on new benchmarks produced from large LTL
specifications.

\paragraph{Related Work.}
In~\cite{FiliotJR10,FiliotJR11}, compositional algorithms are proposed for the
LTL realizability problem. The LTL formulas considered there are assumed to be
conjunctions of smaller LTL formulas, and so the structure of the specification
is directly available to them, while in our case it has to be recovered.  Also,
the main data-structures used there are based on antichains while we use BDDs.
In symbolic model checking algorithms, partitioned transition
relations~\cite{burch1991symbolic} are widely used whenever the system is made
of several components. Here, the goal is to compute the one-step successor
states without explicitly computing the conjunction of the transition relations
for each component.  The image computation is rather done using
\emph{quantification scheduling} heuristics which tries to apply variable
quantification as early as possible inside the conjunction; see \eg
\cite{wang2003compositional}.  We also use partitioned transition relations in
our algorithms: the next-state function for each latch is stored separately.
Unlike forward model checking algorithms, synthesis algorithms proceed
backwards, so we can use the \emph{composition} operation provided by BDD
libraries to compute predecessors, and we do not need any early quantification
heuristics.

\paragraph{Structure of the paper.}
In Section~\ref{sec:prelim}, we fix notation and recall the definitions needed to
present our results. Then, in Section~\ref{sec:decomp}, we describe the class of
decompositions our algorithms accept as input, we give some examples of how to
decompose a succinct safety specification given by an extended AIGER file and
outline the algorithm we implemented to get such a decomposition. Our algorithms
are described in detail in Section~\ref{sec:algos} and the results of our tests
are presented in Section~\ref{sec:experiments}.
      
\section{Preliminaries}\label{sec:prelim}
Let $\mathbb{B} = \{0,1\}$. Given a set of variables~$A$, 
a \emph{valuation over~$A$} is an element of~$\mathbb{B}^A$,
and a set of valuations over~$A$ is represented by its characteristic function
$f : \mathbb{B}^A \rightarrow \mathbb{B}$. We will write~$f(A)$
to make the dependency on the variables~$A$ explicit.
Given two disjoint sets of variables $A,B$, let us write
$\mathbb{B}^{A,B}$ for $\mathbb{B}^A \times \mathbb{B}^B$.
Consider variable sets $A\subseteq B$.
We define the \emph{projection} of a valuation~$v : \mathbb{B}^B$
to~$A$ as $v\downarrow_A : \mathbb{B}^A$, with $v\downarrow_A(a) =1$
if, and only if $v(a) = 1$.
We extend this notation to 
functions $f : \mathbb{B}^B \rightarrow \mathbb{B}$ by
$f \downarrow_A : \mathbb{B}^A \rightarrow \mathbb{B}$, 
defined as $f\downarrow_A(v)$ if, and only if $\exists v' \in \mathbb{B}^B,
f(v')$, and $v = v' \downarrow_A$.
We define the \emph{lifting} of a set~$f: \bool^A \rightarrow \bool$
in~$\bool^B$ by $f\uparrow_B(v) = 1$ if, and only if $f(v\downarrow_A) =1$.
For a set of variables~$A=\{a_1,a_2,\ldots\}$, let us
write~$A'=\{a_1',a_2',\ldots\}$ the set of \emph{primed variables}.  For~$f(A)$,
let $f(A')$ denote the characteristic function $f(A)$ where each variable $a \in
A$ has been renamed as its primed copy $a' \in A'$.

\paragraph{Symbolic Games.}
We formalize the reactive synthesis problem as a two-player turn-based game
with safety objective described symbolically. We consider games
defined by sequential synchronous circuits, 
encoded in the AIGER format.
More precisely, a \emph{game} is a tuple $G = \langle L, X_u, X_c, (f_l)_{l\in
L}, \out\rangle$, where:

\begin{enumerate}
\item $X_u, X_c, L$ are finite disjoint sets of Boolean variables representing
	\emph{uncontrollable inputs}, \emph{controllable inputs}, and
	\emph{latches} respectively;
\item for each latch $l\in L$, $f_l \colon \mathbb{B}^L \times \mathbb{B}^{X_u}
	\times \mathbb{B}^{X_c} \to \mathbb{B}$ is the \emph{transition
	function} that gives the valuation of~$l$ in the next step. In practice
	these functions will be
	given by And-Inverter Graphs (see below for a definition).
\item $\out \in L$  is a distinguished latch which indicates whether an error
	has occurred.  We will often modify the circuit by replacing~$\fbad$ by
	some other Boolean function~$e$, which we denote by $G[\fbad \leftarrow e]$.
\end{enumerate}

A \emph{state}~$q$ of game $G$ is a valuation of latches, that is an element of
$\mathbb{B}^{L}$.  A \emph{valuation}~$v$ in game $G$ is a valuation of latches
and inputs, that is an element of $\mathbb{B}^{L,X_u,X_c}$.  We denote the
\emph{global transition function} $\delta \colon  \mathbb{B}^L \times
\mathbb{B}^{X_u} \times \mathbb{B}^{X_c} \to \mathbb{B}^L$ such that
$\delta(v)(l) = f_l(v)$ for each latch~$l$. 
An \emph{execution} from valuation~$v$ of the game $G$ is a sequence of
valuations $(v_i)_{i \in \mathbb{N}} \in \left(\mathbb{B}^{L, X_u,
X_c}\right)^\omega$ such that $v_0 = v$ and for all $i$,
$$
	v_{i+1}\downarrow_L =
	\delta(v_i \downarrow_L ,v_i \downarrow_{X_u},v_i \downarrow_{X_c}).
$$
The execution is \emph{safe} if, for all $i \ge 0$, we have that $v_i(\out) =
0$.

Note that symbolic games define game arenas of exponential size but we will only
work on their symbolic representations.

\paragraph{Controller synthesis.}
The goal of \emph{controller synthesis} is to find a strategy to determine the
controllable inputs given uncontrollable inputs and the current state (\ie,
valuation of the latches) to ensure
that the error state is not reachable.  A \emph{strategy} is a
function~$\lambda\colon \mathbb{B}^{L, X_u} \to \mathbb{B}^{X_c}$.  An execution
$(v_i)_{i \in \mathbb{N}}$ is \emph{compatible} with $\lambda$ if for all $i \in
\mathbb{N}$, 
\[
	v_i\downarrow_{X_c} = \lambda(v_i\downarrow_{L},v_i\downarrow_{X_u}).
\]
A strategy $\lambda$ is
\emph{winning} if all executions that are compatible with $\lambda$ are safe.  A
valuation $v$ is \emph{winning} if there exists a strategy $\lambda$ that is winning
from $v$.  We denote $W(L,X_u,X_c)$ the \emph{winning valuations} of $G$, that is
the set of valuations that are winning.

\paragraph{And-Inverter Graphs.}
An \emph{And-Inverter Graph} (AIG) is a directed acyclic graph with two-input
nodes representing logical conjunction (AND gates), terminal nodes representing
inputs, and edges that are possibly \emph{inverted} to denote logical negation
(NOT gate). Formally, an AIG is a tuple $G = \langle V,E,\iota \rangle$ such that
$(V,E)$ is a directed graph with every vertex having $0$ or $2$ outgoing edges,
and $\iota : E \to \mathbb{B}$ labels inverted edges with $1$. We depict edges
(not) labelled by $\iota$ as arrows (not) marked with a dark dot.
Figure~\ref{fig:example-aig-1} shows a simple AIG with Boolean variables
$x_1,x_2,x_3,x_4$. Each node in the AIG defines a Boolean function. For example,
$v_1$ defines the Boolean function
$\phi_{v_1} \equiv x_1 \land \lnot \phi_{v_2}$, where
$\phi_{v_2}$ is the corresponding formula defined by $v_2$, since the edge from
$v_1$ to $v_2$ is marked as inverted.

The AIGER format~({\tt http://fmv.jku.at/aiger/FORMAT}) was defined as a standard
file format to describe sequential synchronous circuits (the logic defined as an
AIG), and has been used in model checking and synthesis competitions. In the
latter case, the inputs are partitioned into \emph{controllable} and
\emph{uncontrollable}~({\tt
http://www.syntcomp.org/wp-content/uploads/2014/02/Format.pdf}). This is the
format that we will assume as representation of the input game for our
algorithms.  We call an \emph{AIG game}, a symbolic game described in the AIGER
format.

\paragraph{Binary Decision Diagrams.}
Internally, our tool uses \emph{binary decision diagrams} (BDD)~\cite{bryant86}
to represent Boolean functions used to represent sets of states or (parts of)
transition relations.
We use classical operations and notation on BDDs and refer the interested reader
to~\cite{Andersen97anintroduction} for a gentle introduction to BDDs\@. 
Projection and lifting of functions are easily implemented with BDDs:
projecting is done by an existential quantification and lifting is a trivial
operation because it only extends the domain of the function but its logical
representation, \ie\ its Boolean formula, stays the same.

In our algorithms, we often use BDD operations which implement heuristics to
reduce the size of the given BDD\@,
namely, \emph{generalized cofactors}~\cite{tslbs-v90,shs-vb94}. 
A generalized
cofactor $\hat{f}(X)$ of $f(X)$ with respect to $g(X)$ yields a BDD 
that matches~$f(X)$ inside~$g(X)$, and is defined arbitrarily outside~$g(X)$.
This degree of freedom outside~$g(X)$ allows heuristics to reduce the BDD size.
We write $\hat{f}(X) = \gencof{f(X)}{g(X)}$. Formally, we have that $\hat{f}(X) \land g(X)
= f(X) \land g(X)$ and $\hat{f}$ has at most the size of $f$. BDD libraries
implement the operations \emph{restrict} or \emph{constrain} (see,
\eg~\cite{somenzi99}), which are specific generalized cofactors.

\paragraph{Classical Algorithms to Solve Safety Games.}
We recall the basic fixpoint computation for solving safety games, applied here
on symbolic safety games.  Let $G = \langle L, X_u, X_c, (f_l)_{l\in
L},\out \rangle$ be a symbolic game.
The complement of the set $W(L,X_u,X_c)\downarrow_L$ can be computed by
iterating an \emph{uncontrollable predecessors} operator.  For any set of states
$S(L)$, the \emph{uncontrollable predecessors} of~$S$ is defined as
\[
	\upre_G(S) = \{q \in \mathbb{B}^L \st 
		\exists x_u \in \mathbb{B}^{X_u}.\
		\forall x_c \in \mathbb{B}^{X_c} :
		\delta(q, x_u, x_c) \in S\};
\]
the dual \emph{controllable predecessors} operator is defined as
\[
	\cpre_G(S) = \{q \in \mathbb{B}^L \st 
		\forall x_u \in \mathbb{B}^{X_u}.\
		\exists x_c \in \mathbb{B}^{X_c} :
		\delta(q, x_u, x_c) \in S\};
\]
We denote by $\upre_G^*(S) = \mu X. (S \cup \upre_G(X))$, the \emph{least
fixpoint} of the function $F : X \to S \cup \upre_G(X)$ in the $\mu$-calculus
notation (see \cite{ej91}). Note that $F$ is defined on the powerset lattice,
which is finite. It follows from Tarski-Knaster theorem~\cite{tarski55} that,
because $F$ is monotonic, the fixpoint exists and can be computed by iterating
the application of $F$ starting from any value below it, \eg the least value of
the lattice. Similarly, we denote by $\cpre_G^*(S) = \nu X. (S \cap
\cpre_G(X))$, the \emph{greatest fixpoint} of the function $F : X \to S \cap
\cpre_G(X))$. Dually, we have that, because $F$ is monotonic, the fixpoint
exists and can be computed by iterating the application of $F$ starting from any
value above it, \eg the greatest value of the lattice.
When $G$ is clear from the context, we simply write $\upre$ ($\cpre$)
instead of $\upre_G$ ($\cpre_G$). The Proposition follows from well-known
results about the relationship between safety games and these operators (see,
\eg,~\cite{ag11}).
\begin{proposition}
	For any symbolic game $G = \langle L, X_u, X_c, (f_l)_{l \in L},
	\out\rangle$, we have
	\begin{itemize}
		\item $\cpre^*( (\out \mapsto 0)\uparrow_L) = \cpre(
			W(L,X_u,X_c)\downarrow_L)$; dually,
		\item $\upre^*( (\out \mapsto 1)\uparrow_L) = {\lnot
			\cpre(W(L,X_u,X_c) \downarrow_L)} = {\upre(\lnot
			W(L,X_u,X_c)\downarrow_L)}$.
	\end{itemize}
\end{proposition}

In the rest of the paper, we assume a black-box procedure $\solve$ which, for a
given symbolic game, computes the corresponding winning valuations.  In
practice, $\solve$ can be implemented using $\upre$ or $\cpre$. Formally, 
\[
	\solve(G) = \{ (q,x_u,x_c) \in \mathbb{B}^{L,X_u,X_c} \st q(\out) = 0
	\land \delta(q,x_u,x_c) \not\in \cpre^*( (\out \mapsto 0){\uparrow_L})
	\}.
\]
Note that $\solve$ gives the set of winning valuations, and not the set of
winning states. The interpretation of $\solve(G)$ is that it is the maximal
permissive strategy: any strategy for the controller that ensures to stay within
this set is a winning strategy.  We also consider procedure $\solves(G) = \{ q
\in \mathbb{B}^L \st q \in \cpre^*\left( (\out \mapsto 0){\uparrow_L}
\right)\}$ which returns the set of winning states.

\paragraph{Optimizations Using Generalized Cofactors.}
Let us now establish the correctness of two optimizations we use in the sequel.

We first formalize the dependence on latches as follows.  The \emph{cone of
influence} (see, \eg,~\cite{cgp01}) of~$e_i$, written $\texttt{cone}(e_i)$, is
the set of variables on which~$e_i$ depends, that is, $\texttt{cone}(\Phi)
\subseteq L \cup X_u \cup X_c$ is the minimal set of variables such that if $x
\in \texttt{cone}(\Phi)$ then either $(\exists x: \Phi) \not \Leftrightarrow
\Phi$ or $x \in \texttt{cone}(f_{y})$ for some $y \in \texttt{cone}(\Phi) \cap
L$. For convenience, we denote by $\cone(\Phi)$ the set $\texttt{cone}(\Phi)
\cap L$.

Observe that we have defined the cone of influence of a Boolean function
semantically. That is to say, a variable $x$ is in the cone of
influence of a function $\Phi$ if and only if the set of valuations satisfying
$\Phi$ changes for some fixed valuation of $x$. Since we consider functions
given by AIGs, the cone of influence can be over-approximated by
exploring the AIG starting from the vertex corresponding to function~$\Phi$,
adding all latches and inputs visited and the cones of influence of the latches
-- computed recursively. In our implementation we use this over-approximation
when working on the AIG only and we use the definition on the semantics to
obtain an algorithm on BDDs -- which we use when working with BDDs.

Given an over-approximation $\Lambda$ of the winning valuations
\begin{inparaenum}[$(i)$]
\item we first simplify the transition relation and keep it precise only in
	$\Lambda$,
\item we further modify the transition relation by making every transition not
	allowed by $\Lambda$ go to an error state, \ie change
	$\fbad$.
\end{inparaenum}
In fact, correctness of the first optimization requires that the second one be
used as well. The following result summarizes the properties of these
optimizations.
\begin{lemma}\label{lem:correct-restrict}
	For any symbolic game $G = \langle L, X_u, X_c, (f_l)_{l \in L}, \out
	\rangle$, and any $\Lambda(L,X_u,X_c) \supseteq W(L,X_u,X_c)$, if we
	write $f'_l = \gencof{f_l}{\Lambda}$ for all $l \in L$, we have
	\[
		\solve(G) = 
		\solve(\langle \cone(\Lambda), X_u, X_c, (f'_l)_{l \in
		\cone(\Lambda)}\rangle[\fbad' \leftarrow \lnot \Lambda])
			\uparrow_L.
	\]
\end{lemma}
\begin{proof}
  We first show that solving the game with error function~$\lnot \Lambda$ yields
  the same winning valuations as for~$\fbad$.  For that we will use two basic
  properties of the winning valuations: first if $f \subseteq f'$ then
  \[
	  \solve(G[\fbad \leftarrow f']) \subseteq \solve(G[\fbad \leftarrow f]);
  \]
  secondly
  \[
	  \solve(G[\fbad \leftarrow \lnot \solve(G)]) = \solve(G),
  \]
  this is because if an execution compatible with strategy~$\lambda$ reaches
  $\lnot \solve(G)$, then by definition of winning valuations it can be extended
  from there to an execution compatible with $\lambda$ that is unsafe. Together
  with the fact that $\fbad \subseteq \lnot\Lambda \subseteq \lnot W$, these
  properties imply that $\solve(G) = \solve(G[\fbad \leftarrow \lnot \Lambda])$.
  It is clear that one can consider only the variables in $\cone(\Lambda)$ for
  this computation, and thus considering $H = (\langle \cone(\Lambda), X_u, X_c,
  (f_l)_{l \in \cone(\Lambda)}\rangle [\fbad \leftarrow \lnot \Lambda])$, we
  have
  \[
	  \solve(G) = \solve(H) \uparrow_L.
  \]

  It remains to show that the same set~$\solve(H)$ is obtained when the
  functions $f_l'$ are used transition functions~$f_l$.  Let us denote $G' =
  (\langle \cone(\Lambda), X_u, X_c, (f_l')_{l \in \cone(\Lambda)}\rangle
  [\fbad' \leftarrow \lnot \Lambda])$.  We note that, for any $u \supseteq \lnot
  \Lambda$, the following holds:
  \[
	  \upre_G'(u) \cup u = \upre_H(u) \cup u.
  \]
  Hence, it is straightforward to show by induction that $\solve(H) =
  \solve(G')$.
\end{proof}

\section{Decomposing the Specification}\label{sec:decomp}
In this section, we describe how we decompose the error function~$\fbad$
of a given symbolic game into a disjunction \textit{\ie}
$\fbad \equiv \big(\bigvee_{1 \leq i \leq n} e_i\big)$.
Notice that if a strategy $\lambda(L, X_u, X_c)$ ensures
that $\fbad$ is never true then it also ensures that $e_i$
is never true.  We will then give algorithms that solve the game where
each~$e_i$ is seen as the error function, and combine the obtained solutions
into a global solution. 

The rationale behind this approach is that the functions~$e_i$ do not depend on
all latches in general, so solving the game for~$e_i$ is often efficient.

\paragraph{Sub-game.}
Given a decomposition of $\fbad$, we define a \emph{sub-game} $G_i$ by replacing
the error function by $e_i$ and considering only variables in its cone of
influence. Formally, we write 
\[
	G_i = \langle \cone(e_i), X_u, X_c, (f_l)_{l \in \cone(e_i)}
	\rangle[\fbad \leftarrow e_i].
\]
We will often use the notation
$G[\fbad \leftarrow e_i]$, which consists in replacing the function~$\fbad$
by~$e_i$.  In practice, the size of the symbolic representation of the sub-games
are often significantly smaller than that of the original game. Recall also that
winning all the sub-games is necessary to win the global game.  We write
$W_i(\cone(e_i),X_u,X_c)$ for the winning valuations of $G_i$.
In the implementation, $S_i$ and $S_i \uparrow_L$ are represented by the same BDD.

\begin{figure}
\centering
\begin{tikzpicture}[edge from parent/.style={draw,-latex}]
	\node[vertex] (root) {$v_1$}
	child {	node {$x_1$} }
	child {	node[vertex] {$v_2$}
		child {	node {$x_2$} }
		child {
			node[vertex] {}
			child{%
				node {$x_3$}
				edge from parent node[inverted]{}
			}
			child {node {$x_4$} }
		}
		edge from parent node[inverted]{}
	};
\end{tikzpicture}
\caption{Example AIG}
\label{fig:example-aig-1}
\end{figure}
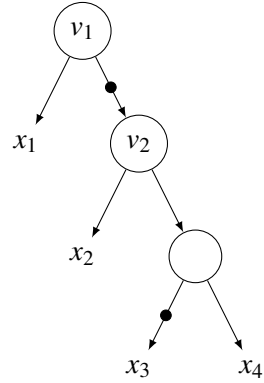

\paragraph{Example 1.} Consider the AIG shown in
Figure~\ref{fig:example-aig-1} where $x_1,x_2,x_3,x_4$ are all input variables.
We would like to decompose the function defined by the sub-tree rooted at $v_1$
(\ie\ the whole tree) which we will denote by $\phi_{v_1}$. It should be clear that
\(
	\phi_{v_1} \equiv x_1 \land \lnot \phi_{v_2}
\)
where $\phi_{v_2}$ is the function defined by the sub-tree rooted at $v_2$. In
turn, we also have that
\(
	\phi_{v_2} \equiv x_2 \land \lnot x_3 \land x_4.
\)
If we distribute the disjunction from $\lnot f_{v_2}$ we get that
\(
	\phi_{v_1} \equiv (x_1 \land \lnot x_2) \lor
		(x_1 \land x_3) \lor 
		(x_1 \land \lnot x_4).
\)
Thus, one possible decomposition of $\phi_{v_1}$ would be to take $e_1 = x \land
\lnot x_2$, $e_2 = x_1 \land x_3$, and $e_3 = x_1 \land \lnot x_4$.

The general steps followed in Example $1$ above can be generalized into an
algorithm which outputs a decomposition of the error function whenever one
exists. Intuitively, the algorithm consists in exploring all non-inverted edges
of the AIG graph from the vertex which defines the error function. If there are
no inverted edges which stopped the exploration, or if all of them lead to
leaves, the error function is in fact a conjunction of Boolean variables and can
clearly not be decomposed. Otherwise, there is at least one inverted edge
leading to a node representing an AND gate.  In this case, we can push the
negation one level down and obtain a disjunction which can be distributed to
obtain our decomposition. Algorithm~\ref{alg:decompose} details the procedure we
have implemented. It takes as input an AIG, whether the error function is
inverted, and the vertex $v_\out$ which defines the error function. It outputs a
set of functions whose conjunction is logically equivalent to the error
function.

We have kept our description of Algorithm~\ref{alg:decompose} and
Algorithm~\ref{alg:get-minput-and} (called by the former) informal.

\begin{algorithm}
	\small
	$to\_visit$ := $\{v_0\}$\;
	$pos$ := $\{\}$\;
	$neg$ := $\{\}$\;
	\While{$|to\_visit| > 0$}{
		Pop $u \in to\_visit$\;
		\eIf{$u$ is not a leaf}{%
			Let $e =(u,v)$ and $e' = (u,v')$ be s.t. $e,e' \in E$\;%
			\eIf{$\iota(e) = 1$}{%
				$neg$ := $neg \cup \{v\}$;%
			}{%
				$to\_visit$ := $to\_visit \cup \{v\}$;%
			}%
			\eIf{$\iota(e') = 1$}{%
				$neg$ := $neg \cup \{v\}$;%
			}{%
				$to\_visit$ := $to\_visit \cup \{v\}$;%
			}%
		}{%
			$pos$ := $pos \cup \{u\}$;%
		}%
	}
	\Return{$(pos,neg)$}
\caption{$\texttt{get\_minput\_and}(V,E,\iota,v_0)$}
\label{alg:get-minput-and}
\end{algorithm}

\begin{algorithm}
	\small
	$(pos, neg)$ := $\texttt{get\_minput\_and}(V,E,\iota,v_\out)$\;
	\If{$inv = 1$}
	{
		\Return{ $\{\lnot \phi_v \st v \in pos\} \cup \{\phi_v \st v \in neg\}$ }
	}
	\If{$inv = 0$ and all $v \in neg$ are leaves}
	{
		\Return{ $\{ \phi_{v_\out} \}$;
		\tcc*[f]{\scriptsize No decomposition possible}}
	}
	Take $v_0 \in \argmax\{||\texttt{get\_minput\_and}(V,E,\iota,v)|| \st
v \in neg\}$; \tcc*[f]{\scriptsize where $||(S_1,S_2)||  := |S_1| +
|S_2|$}\\
	$res$ := $\bigwedge_{u \in pos} \phi_u \land \bigwedge_{v \in neg
		\setminus \{v_0\}} \lnot \phi_v$\;
	$(pos, neg)$ := $\texttt{get\_minput\_and}(V,E,\iota,v_0)$\;
	\Return{ $\{res \land (\lnot \phi_v) \st v \in pos\} \cup \{ res \land \phi_v \st v \in neg\}$ }
\caption{$\decompose(V,E,\iota,inv, v_\out)$}
\label{alg:decompose}
\end{algorithm}

\paragraph{Example 2.} Consider a formula given by a set of assumption
formulas $\{A_i(L,X_u) \st 1 \le i \le n\}$ and a set of guarantees 
$\{G_j(L,X_u,X_c) \st 1 \le j\le m\}$.\footnote{This is actually the way in
	which the error formula is stated for, \eg, the \textsf{AMBA}
	benchmarks.}
The system we want to synthesize is expected to determine the
controllable inputs in way such that if the assumptions are true, then the
guarantees are met. This is formally stated as Equation~\ref{eqn:ass-gua}.
\begin{equation}\label{eqn:ass-gua}
	\Phi = \big(\bigwedge_{1 \le i \le n} A_i\big) \implies
	\big( \bigwedge_{1 \le j \le m} G_j \big)
\end{equation}
A natural decomposition for the error function $\lnot \Phi$ would be the
following: $\bigvee_{1 \le j \le m} \big(\lnot G_j \land \bigwedge_{1 \le i \le
n} A_i \big)$.  If $\lnot \Phi$ were given as the AIG depicted in
Figure~\ref{fig:example-aig-2}, then it is not hard too see that
Algorithm~\ref{alg:decompose} would yield a very similar
decomposition.
Indeed, as we have not assumed anything in particular about the formulas $A_i$
and $G_j$ we cannot tell whether Algorithm~\ref{alg:get-minput-and} will explore
beyond each $G_j$, thus giving us more sub-games than the proposed
decomposition.  However, in practice, this is even better as smaller sub-games
usually depend on less variables. This, in turn, could lead to them being easier
to solve.

\begin{figure}
\centering
\begin{tikzpicture}[edge from parent/.style={draw,-latex},
		    level 1/.style={sibling distance=3cm},
	    	    level 2/.style={sibling distance=1.5cm}]
	\node[vertex] (root) {}
	child {	node[vertex] {$a$}
		child {
			node[vertex] {}
			child {
				node {$A_n$}
				edge from parent[dotted]
			}
			child { 
				node {$A_i$}
				edge from parent[solid]
			}
			edge from parent[dotted]
		}
		child { node {$A_1$} }
	}
	child {	node[vertex] {$g$}
		child { node {$G_1$} }
		child {
			node[vertex] {}
			child { node {$G_j$} }
			child {
				node {$G_m$}
				edge from parent[dotted]
			}
			edge from parent[dotted]
		}
		edge from parent node[inverted]{}
	};
\end{tikzpicture}
\caption{One possible AIG for Equation~\ref{eqn:ass-gua}}
\label{fig:example-aig-2}
\end{figure}
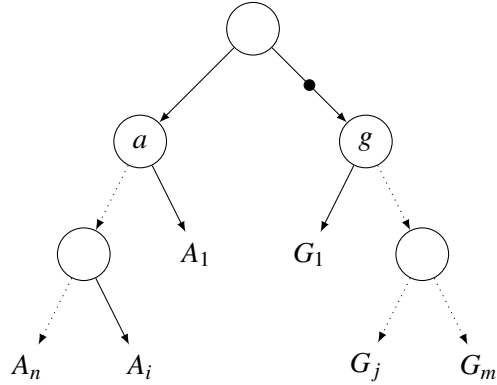

\begin{lemma}\label{lem:nec-allsub}
  For each sub-game~$G_i$ with new error function $e_i$, we have that
  \[ W(L,X_u,X_c) \subseteq (W_i \uparrow_L)(L,X_u,X_c).\]
\end{lemma}
\begin{proof}
  For each valuation $v' \in W(L,X_u,X_c)\downarrow_{\cone(e_i) \cup X_u \cup
  X_c}$, we select a valuation $v \in W(L,X_u,X_c)$.  Let $\lambda_v$ be a
  winning strategy in $G$ from $v$.  Since there is no losing
  outcome for $\lambda_v$, for all $x_u \in \mathbb{B}^{X_u}$,
  $\lambda_v(\delta(v),x_u)$ is such that
  $(\delta(v),x_u,\lambda_v(\delta(v),x_u)) \in W(L,X_u,X_c)$.  For all $x_u \in
  \mathbb{B}^{X_u}$, we fix $\lambda'(\delta(v'),x_u)$ to be
  $\lambda_v(\delta(v),x_u)$.  We have that
  $(\delta(v'),x_u,\lambda'(\delta(v'),x_u)) \in W(L,X_u,X_c)
  \downarrow_{\cone(e_i) \cup X_u \cup X_c}$ because the transition relations of
  $G$ and $G_i$ coincide on $\cone(e_i) \cup X_u \cup X_c$.  The strategy
  $\lambda'$ ensures that any execution which starts in
  $W(L,X_u,X_c)\downarrow_{\cone(e_i) \cup X_u \cup X_c}$ stays inside
  $W(L,X_u,X_c)\downarrow_{\cone(e_i) \cup X_u \cup X_c}$.  Since $e_i$
  evaluates to false on $W(L,X_u,X_c)$, these states are not error states in
  $G_i$. Therefore $\lambda'$ is winning for all states in
  $W(L,X_u,X_c)\downarrow_{\cone(e_i) \cup X_u \cup X_c}$.  This implies that
  $W_i$ contains the projection of all winning states of $G$ and therefore $W
  \subseteq W_i\uparrow_L $.
\end{proof}

\section{Compositional Algorithms}\label{sec:algos}
In this section, we give three algorithms to solve AIG games compositionally.
Each algorithm first solves the sub-games, and then combines
the solutions using different heuristics.  We denote by
$\decompose$ the procedure that implements the decomposition of~$\fbad$
described in Section~\ref{sec:decomp}, and returns the set of error
functions~$e_i$.
In all three algorithms, we start by solving each sub-game and obtaining the
winning valuations~$W_i(L, X_u, X_c)$, for $1 \le i \le n$.
These steps are given in lines $1$--\ref{alg1-loc:endofloop}, and are common to
ll our algorithms; we assume that \solve{} raises an exception and terminates the
program if the sub-game cannot be won. Otherwise, we aggregate the results and
solve the global game; for the latter, we adopt a different approach in each of
the three algorithms.

\subsection{Global aggregation}

\begin{algorithm}
	\small
	$\{e_1,\ldots,e_n\}$ := $\decompose(\fbad)$;
	\tcc*[f]{\scriptsize Formulas $e_i(L,X_u,X_c)$ s.t. 
          $\fbad \equiv \bigvee_{1 \le i \le n} e_i$}\\
	\For{$1 \le i \le n$}
	{
		$w_i(L, X_u, X_c)$ := $\solve(\langle 
		\cone(e_i), X_u, X_c, (f_l)_{l \in \cone(e_i)}[\fbad \leftarrow
    e_i] \rangle){\uparrow_{L,X_u,X_c}}$\;
                \label{alg1-loc:endofloop}
	}
	$\Lambda(L, X_u, X_c)$ := $\bigwedge_{1 \le i \le n}
	w_i(L, X_u, X_c)$\;
	\lFor{$l \in \cone(\Lambda)$}
	{
	  $f_l'(L, X_u, X_c)$ := $\gencof{f_l(L, X_u, X_c)}{\Lambda(L, X_u,
		X_c)}$
	}
	\label{alg1-loc:beforeret}
	\Return $\solve(\langle 
	\cone(\Lambda), X_u, X_c, (f'_l)_{l \in \cone(\Lambda)} \rangle
	     [\fbad' \leftarrow \lnot \Lambda]){\uparrow_{L,X_u,X_c}}$\;
\caption{\texttt{comp\_1}$(\langle L, X_u, X_c, (f_l)_{l\in L}\rangle)$}
\label{alg:algo1}
\end{algorithm}

In Algorithm~\ref{alg:algo1}, we start by computing the intersection of the
winning valuations: $\Lambda = \bigwedge_{1\leq i \leq n} W_i$. In fact, any
valuation that is not in~$\Lambda$ is losing in one of the sub-games; thus in
the global game. Conversely, a strategy that stays in~$\Lambda$ is winning for
each sub-game.  Therefore, we solve the global game with the new safety objective of
avoiding $\lnot \Lambda$.
Before solving the global game, the algorithm attempts to reduce the size of the
transition relations by virtue of Lemma~\ref{lem:correct-restrict}.

\begin{theorem}
  \label{thm:algo1-correct}
  Algorithm~\ref{alg:algo1} computes the winning valuations of the given AIG game.
\end{theorem}	
\begin{proof}
  We prove first that $W \subseteq \Lambda$ (that is for all valuation~$v$,
  $W(v) \Rightarrow \Lambda(v)$). Since $\lnot e_i \supseteq W(L,X_u,X_c)$, we
  get -- by Lem.~\ref{lem:correct-restrict} -- that
  each~$w_i(L,X_u,X_c)$ is $W_i\uparrow_L$ where $W_i$ is the winning valuations of
  the sub-game~$G_i$.  If $q \not\in \Lambda(L, X_u, X_c)$, there is a sub-game
  $G_i$ such that $\pi_i(q)$ is not winning.  By Lem.~\ref{lem:nec-allsub}, this
  implies that $q$ is not winning in $G$, hence $q\not\in W(L, X_u, X_c)$.

  From Lem.~\ref{lem:correct-restrict} it then follows that
  $\solve(G) = \solve(G')\uparrow_L$ and therefore the
  algorithm computes the correct result.
\end{proof}

\subsection{Incremental aggregation}
\begin{algorithm}
	\small
	$\{e_1,\ldots,e_n\}$ := $\decompose(\fbad)$;
	\tcc*[f]{\scriptsize Formulas $e_i(L,X_u,X_c)$ s.t. 
          $\fbad \equiv \bigvee_{1 \le i \le n} e_i$}\\
	\For{$1 \le i \le n$}
	{
		$w_i(L, X_u, X_c)$ := $\solve(\langle 
		\cone(e_i), X_u, X_c, (f_l)_{l \in \cone(e_i)}[\fbad \leftarrow
    e_i] \rangle){\uparrow_{L,X_u,X_c}}$\;
                \label{alg2-loc:endofloop}
	}
	$E$ := $\{ w_i \st 1 \le i \le n\}$\;
	\While{$|E| > 1$}
	{
		$\begin{array}{ll}
			(r,s) := \argmax_{(i,j) \in |E|^2 : i \neq j } &
			  \{{\alpha} \cdot \texttt{bddsize}(\lnot (w_i \land
				w_j))\\
			&+ \beta |\cone(w_i) \cap \cone(w_j)|\\
			&+ \gamma |\cone(w_i) \cup \cone(w_j)|\};
		\end{array}$
		\label{algo2:combination}

		\lFor{$l \in \cone(w_r \land w_s)$}
		{
			$f_l'(L,X_u,X_c)$ := $\gencof{f_l(L,X_u,X_c)}{
			(w_r \land w_s)}$
		}
		$w(L,X_u,X_c)$ := $\solve(\langle \cone(w_r \land w_s),X_u,X_c,
		(f_l')_{l \in \cone(w_r \land w_s)}\rangle[ \fbad \leftarrow\lnot 
		w_r \lor \lnot w_s]){\uparrow_{L,X_u,X_c}}$\;
		\label{line:algo2:solve}
		Remove $w_r,w_s$ and add $w$ to $E$\;
		\label{line:addrs}
	}
	\Return last $w(L,X_u,X_c) \in E$\;
\caption{\texttt{comp\_2}$(\langle L, X_u, X_c, (f_l)_{l\in L}\rangle, \alpha,
\beta, \gamma)$}
\label{alg:algo2}
\end{algorithm}
In Algorithm~\ref{alg:algo2}, we aggregate the results of the sub-games
\emph{incrementally}: given the list of winning valuations $w_i$ for the
sub-games, at each iteration, we choose and remove two sub-games~$i$ and~$j$,
solve their conjunction (as in Algorithm~\ref{alg:algo1}, with error function
$\lnot (w_i \land w_j)$), and add the newly obtained winning valuations back in
the list. 
To choose the sub-games, we use the following heuristics; we assign a score to
each pair of sub-games based on the size of the BDD of the error function $\lnot
(w_i\land w_j)$, and on the number of shared latches, and the number
of the latches that appear in either of the sub-games.  Intuitively, we prefer
to work with small BDDs, and to merge sub-games that share a lot of latches,
while yielding a small number of total latches.  We thus use a linear
combination at line~\ref{algo2:combination} to choose the best scoring pair. In
our experiments, we used $\alpha=-2, \beta = 1, \gamma = -1$.

\begin{theorem}\label{thm:algo2-correct}
  Algorithm~\ref{alg:algo2} computes the winning valuations of the given AIG game.
\end{theorem}	
\begin{proof}
Let us denote by $w^i_1, \dots, w^i_{n_i}$ the content of~$E$ at the beginning
of iteration~$i$.  We define a function $F$ from winning valuations $w^i_j$
to subsets of $\{1,\ldots,n\}$.  Intuitively, $F(w^i_j)$ is the set of sub-games
that were solved to obtain~$w^i_j$. For instance, at the first iteration, if
sub-games~$r,s$ are combined -- and the result, $w$, is added to $E$ -- then we
get $F( w ) = \{r,s\}$. For convenience, we assume that $w$ is appended at the
end of the sequence $w^i_1, \dots, w^i_{n_i}$ at line \ref{line:addrs}.

We proceed by induction on~$i$ to define~$F$.  Initially $F( w_i^1 ) =
\{ i \}$ for all~$1\leq i \leq n$.  For $i>1$, for all~$j \neq r,s$, the
element $w^i_j$ remains in the list so~$F$ is already defined on $w^i_j$.
For the newly element $w^i_{n_i}$ we let $F( w^i_{n_i} ) =
F( w^{i-1}_r ) \cup F(w^{i-1}_s)$.

We claim that at any iteration~$i$, $w^i_j$ is the winning valuations of the game
whose error function is the disjunction of the negation of the
winning valuations of the sub-games in $F(w_j^i)$. More precisely,
\[
	w_j^i = \solve(\langle L, X_u, X_c,(f_l)_{l \in L}\rangle[\fbad \leftarrow
	\bigvee_{k \in F(w^i_j)} e_k]).
\]
The correctness of the algorithm will follow since the sets $F(\cdot)$ are
merged at each iteration, and the algorithm always stops with~$|E|=1$ and $F(w)
= \{1,\ldots,n\}$.

The condition holds initially as shown in Theorem~\ref{thm:algo1-correct}.
Let~$i>1$.  As shown in Lem.~\ref{lem:correct-restrict}, the generalized
cofactor operation applied before the call to \texttt{solve} does not affect the
returned set.  Let us denote $E_r = \bigvee_{k \in F(w^{i-1}_r)} e_k$ and $E_s =
\bigvee_{k \in F(w^{i-1}_s)} e_s$.  Let us write $\mathcal{E}= E_r \lor E_s$.
We have $E_r \Rightarrow \lnot w_r$ by induction, and similarly $E_s \Rightarrow
\lnot w_s$; thus $\mathcal{E} \Rightarrow \lnot w_r \lor \lnot w_s$.  Moreover,
for any $q(L,X_u)$ if the controller plays strategy~$x_u \in \mathbb{B}^{X_c}$
with $\lnot w_r(q,x_u)$, or $\lnot w_s(q,x_u)$, then he loses for the error
function defined by $\mathcal{E}$. In other terms, $\lnot w_r \lor \lnot w_s$ is
a subset of losing valuations for error function~$\mathcal{E}$, and contains
$\mathcal{E}$, the set of states losing in one step.  It follows that
$w(L,X_u,X_c)$ computed at step~\ref{line:algo2:solve} is the winning valuations for
the error function~$\mathcal{E}$.
\end{proof}

\subsection{Back-and-forth}
\begin{algorithm}
\small
	$\{e_1,\ldots,e_n\}$ := $\decompose(\fbad)$;
	\tcc*[f]{\scriptsize Formulas $e_i(L,X_u,X_c)$ s.t. 
          $\fbad \equiv \bigvee_{1 \le i \le n} e_i$}\\
	\For{$1 \le i \le n$}
	{
		$w_i(L,X_u,X_c)$ := $\solve(\langle 
		\cone(e_i), X_u, X_c, (f_l)_{l \in \cone(e_i)}[\fbad \leftarrow
		e_i] \rangle){\uparrow_{L,X_u,X_c}}$\;
		$s_i(L)$ := $\solves(\langle 
		\cone(e_i), X_u, X_c, (f_l)_{l \in \cone(e_i)}[\fbad \leftarrow
		e_i] \rangle){\uparrow_{L}}$\;
                \label{alg3-loc:endofloop}
	}
	$\Lambda(L, X_u, X_c)$ := $\bigwedge_{1 \le i \le n}
	w_i(L, X_u, X_c)$\;
	$G'$ := $\langle \cone(\Lambda), X_u, X_c, (f_l)_{l \in
	\cone(\Lambda)}\rangle[\fbad \leftarrow \lnot \Lambda ]$\;
	$u(L)$ := $0$\;
	$u'(L)$ := $\bigvee_{1 \le i \le n} \lnot s_i(L)$;
	\tcc*[f]{The union of all the losing states}\\
	\While{$u \neq u'$}
	{
		$u(L)$ := $u'(L)$\;
		$u'(L)$ := $u(L) \lor \upre_{G'}(u)$\;
	  \label{line:upre}
		\For{$1 \le i \le n$}
		{
			$p_i(L)$ := $\forall L \setminus \cone(e_i) :
			u'(L)$;
			\tcc*[f]{Universal projection of latches not present in local
			sub-game}\\
			\If{$p_i \land s_i \neq 0$}
			{
				\lFor{$l \in \cone(p_i)$}
				{
					$f_l'(L,X_u,X_c)$ := $\gencof{f_l(L,X_u,X_c)}{
					\lnot p_i(L)}$
				}
				$s_i(L)$ := $\solves(\langle X_u, X_c,
					\cone(p_i), (f'_l)_{l \in
					\cone(p_i)}\rangle[\fbad \leftarrow
					p_i\uparrow_{L,X_u,X_c}]){\uparrow_{L}}$\;
			}
		}
		$u'(L)$ := $u'(L) \lor \lnot \bigwedge_{1 \le i \le n} (s_i(L)
		\downarrow_L)$\;
	}
	\Return $\lnot u(L)$\;
\caption{\texttt{comp\_3}$(\langle L, X_u, X_c, (f_l)_{l\in L}\rangle)$}
\label{alg:algo3}
\end{algorithm}

In Algorithm~\ref{alg:algo3}, we interleave the analysis of the global game
(with objective $\Lambda$) and the analysis of the sub-games. At each iteration,
we extend the losing states $u(L)$ by one step, by applying once the $\upre$
operator.  We then consider each sub-game, and check whether the new set $u'(L)$
of losing states (projected on the sub-game), changes the local winning states.
Here, $p_i(L)$ is this projection on the local state-space of sub-game $i$.  We
update the strategies~$\lambda_i$ of the sub-games when necessary, and restart
until stabilization.  Because analyzing the sub-games is often more efficient
than analyzing the global game, this algorithm improves over
Algorithm~\ref{alg:algo1} in some cases (see the experiments' section).  A
similar idea was used in \cite{FiliotJR11} for the problem of synthesis from LTL
specifications.

\begin{theorem}
  Algorithm~\ref{alg:algo3} computes the winning valuations of the given AIG game.
\end{theorem}
\begin{proof}
    Let~$W(L)$ denote the set of winning states of the game~$G$.
    We consider the following invariant.
    \begin{equation}
    \begin{array}{l}
     \forall i\in\{1,\ldots,n\}, W(L) \subseteq s_i(L),
     \\
      \out \subseteq u'(L) \subseteq \lnot W(L).
    \end{array}
    \label{eqn:invar3}
    \end{equation}
    In words, in every iteration, 
    $u'(L)$ is contained in the losing valuations of the global game, and
    each $s_i(L)$ contains the winning valuations of~$G_i$.

    Initially, by Algorithm~\ref{alg:algo1}, $W \subseteq s_i(L)$
    for all~$i$, and we have
    $\out \subseteq \lnot s_i(L)$. So $\out \subseteq \lnot \land_i s_i(L)$.
    Thus, $\out \subseteq u'(L)$. 
    Moreover, since $\lor_i \lnot s_i(L) \subseteq \lnot W$, we have
    that $u'(L) \subseteq \lnot W$.

    Consider now iteration~$i>1$, and assume the invariant holds at the
    beginning. $u'(L)$ is updated at line~\ref{line:upre}.
    The property $\out \subseteq u'(L) \subseteq \lnot W$ still holds by the
    definition of the $\upre$ operation, and by the fact that the set $u'$ 
    can only grow at this step (because of the union).

    We consider now the for loop, and show that $W \subseteq s_i$ after each
    iteration. Assume $p_i \cap s_i \neq \emptyset$ since otherwise~$s_i$ is not
    modified. By definition $p_i \subseteq u'$
    thus $p_i \subseteq \lnot W$. Then the \texttt{solve} function
    computes the set of states from which the controller can avoid 
    $\out \lor p_i$. Since $\out \lor p_i \subseteq \lnot W$, we get that $s_i
    \subseteq W$.
    It follows that $\lnot \bigwedge_{i=1}^n s_i \subseteq \lnot W$.
    Thus, at the last line of the while loop, we have $\out \subseteq u'(L)
    \subseteq \lnot W$.

    Now, line~\ref{line:upre} ensures that after iteration~$i$, $u'(L)$ contains 
    the $i$-th iteration of the $\upre$ fixpoint computation. Hence, the test $u\neq
    u'$ of the while loop ensures that the while loop terminates with~$u(L)$
    being equal to~$\upre^*(G)$.
\end{proof}

\section{Experiments}\label{sec:experiments}
We implemented our algorithms in the synthesis tool AbsSynthe~\cite{bprs14}. We
compare their running times against the most efficient algorithm of AbsSynthe
that implements a backward fixpoint algorithm.\footnote{The new version of
	AbsSynthe with the implementation of the compositional algorithms can be
	fetched from \url{https://github.com/gaperez64/abssynthe}.}
This algorithm was the winner of the \textbf{$2014$ Synthesis Competition}
synthesis track, and the winner of the realizability track at the same
competition implemented a similar backward algorithm.

Let us first illustrate the advantage of the compositional approach
with two examples. In the first set of benchmarks we consider, the controller is to
compute the multiplication of two Boolean matrices given as (uncontrollable) input.
Since each cell of the resulting matrix
depends only on a subset of inputs, namely, on one row and one column, these
benchmarks are well adapted for compositional algorithms. 
Figure~\ref{fig:mult} compares the performances of the classical algorithm with
Algorithm~\ref{alg:algo1}. The classical algorithm was able to solve 36 instances,
while the compositional algorithm solved all 75 instances and was significantly
faster.
The x-axis shows the number of solved benchmarks
within the running time given by the y-axis.  
The second set of benchmarks we consider consist in a washing system made of~$n$
tanks. An uncontrollable input can request at any time the tank to be activated,
at which point the controller should fill the tank with water, and empty it 
after at least~$k$ steps. Moreover, some subsets of tanks cannot be filled at
the same time, and a light is to be on if at least one tank is filled with water.
Note that the control strategy for each tank is not independent due to mutual
exclusion constraints, and to the light indicator.
Algorithm~\ref{alg:algo1} was also more efficient on these benchmarks, as shown
on Fig.~\ref{fig:sched}. The classical algorithm solved 132 benchmarks out of 256, 
while Alg.~\ref{alg:algo1} solved 152.

\begin{figure}[H]
  \begin{minipage}{0.43\textwidth}
    \begin{center}
      \includegraphics[width=1\textwidth]{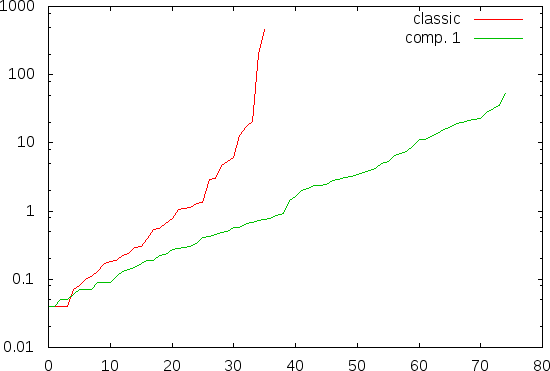}
      \caption{Performances for 75 Boolean matrix multiplication
        benchmarks for Algorithm~\ref{alg:algo1} and the classical algorithm.
      }
      \label{fig:mult}
    \end{center}
  \end{minipage}
  \qquad
  \begin{minipage}{0.43\textwidth}
    \begin{center}
      \includegraphics[width=1\textwidth]{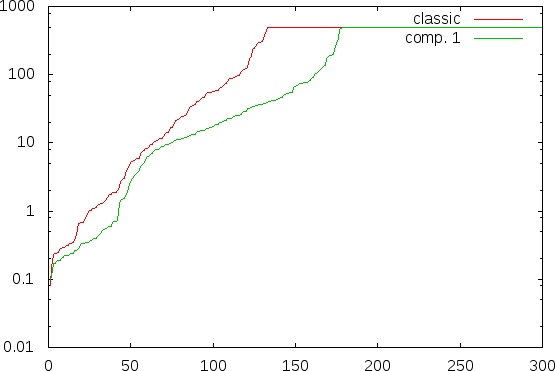}
      \caption{Peformances for the 256 washing system benchmarks
      for Algorithm~\ref{alg:algo1} and the classical algorithm.}
      \label{fig:sched}
    \end{center}
  \end{minipage}
\end{figure}

We now evaluate all three compositional algorithms and compare them with the classical
algorithm on a large benchmark set of 674 benchmarks.
$562$ of these benchmarks were provided for the \textbf{$2014$ Synthesis Competition}
and $105$ have been generated by the new version of LTL2AIG~\cite{ltl2aig} which
translates conjunctions of LTL specifications into AIG.\footnote{A collection of benchmarks,
including the ones mentioned here, can be fetched from
\url{https://github.com/gaperez64/bench-syntcomp14}
and \url{https://github.com/gaperez64/bench-ulb-syntcomp15}.}
Among those benchmarks, $351$ are decomposable by our static analysis into at
least $2$ smaller sub-games. More specifically, the average number of sub-games our
decomposition algorithm outputs is $29$; the median is $21$.

In general, the performances of the three compositional algorithms can differ,
but they are complementary.
Figures~\ref{fig:load} to~\ref{fig:amba} show the performances of the algorithms
on several sets of benchmarks. All benchmarks in
Figures~\ref{fig:load} and~\ref{fig:gb} are decomposable.  Figure~\ref{fig:best}
shows all the benchmarks we used and Figure~\ref{fig:amba} shows only those
benchmarks from last year's synthesis competition which were based on
specifications of the \textsf{AMBA} arbiter.

\paragraph{{\bf Conclusion.}} Even if AIG synthesis problems are monolithic, the
experiments show that our compositional approach was able to solve problems that
can not be handled by the monolithic backward algorithm; our compositional
algorithms are sometimes much more efficient. There are also examples that can
be decomposed but which are not solved more efficiently by the compositional
algorithms. So, it is often a good idea to apply all the algorithms in parallel.
This portfolio approach improved the performance and was able to solve 20
benchmarks that could not be solved by
the fastest algorithms of last year's reactive synthesis competition.

\begin{figure}[H]
  \begin{minipage}{0.49\textwidth}
    \begin{center}
      \includegraphics[width=1\textwidth]{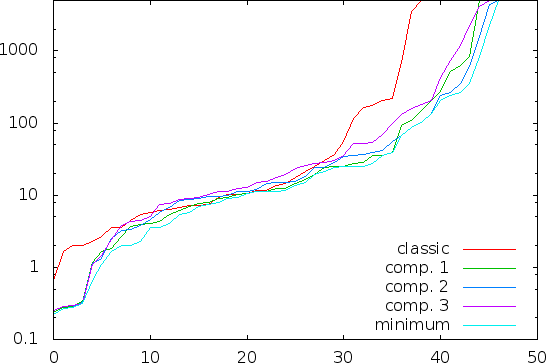}
      \caption{{Performances for 68 load-balancing benchmarks translated from LTL.
        The classical algorithm solves 38 benchmarks, comp.1 44, comp.2 45, comp.3 45.
        In total there are 46 benchmarks that can be solved. The largest example that can be solved has 4005 latches and the smallest example that cannot be solved has 670 latches.}
      }
      \label{fig:load}
    \end{center}
  \end{minipage}
  \begin{minipage}{0.49\textwidth}
    \begin{center}
      \includegraphics[width=1\textwidth]{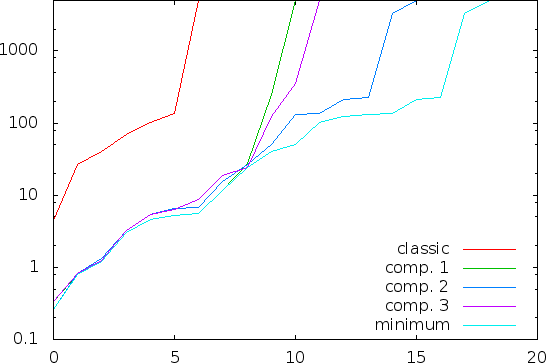}
      \caption{{Performances for 46 generalized buffer benchmarks translated from LTL.
        The classical algorithm solves 6 benchmarks, comp.1 10, comp.2 15, comp.3 11.
        In total there are 18 benchmarks that can be solved. The largest example that can be solved has 22662 latches and the smallest example that cannot be solved has 590 latches.}
      }
      \label{fig:gb}
    \end{center}
  \end{minipage}
  \begin{minipage}{0.49\textwidth}
    \begin{center}
      \includegraphics[width=1\textwidth]{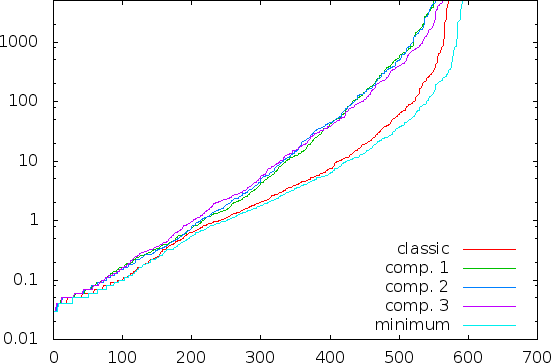}
      \caption{{Performances for the 674 benchmarks.
        The classical algorithm was able to solve 572 benchmarks.
        20 more benchmarks were solved by one of the three compositional algorithms.}
      }
      \label{fig:best}
    \end{center}
  \end{minipage}~\hfill~\begin{minipage}{0.49\textwidth}
    \begin{center}
      \includegraphics[width=1\textwidth]{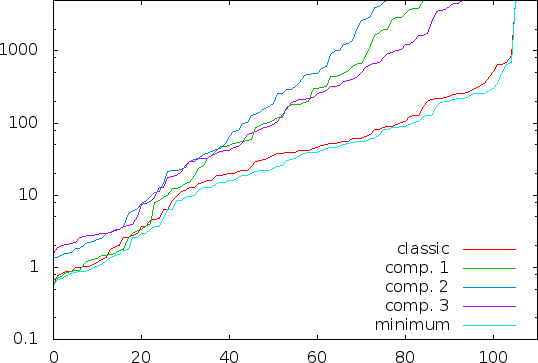}
	\caption{{Performances for 108 AMBA benchmarks.
          The classical algorithm was able to solve 106 benchmarks, comp.1 84, comp.2 76, comp.3 93.}}
	\label{fig:amba}
    \end{center}
  \end{minipage}\\
\end{figure}

\bibliographystyle{eptcs}
\bibliography{refs}

\end{document}